\documentclass[conference]{IEEEtran} 

\usepackage{color}
\usepackage{epsfig}
\usepackage{amsmath}
\usepackage{amsfonts}
\usepackage{amsthm}
\usepackage{amssymb}
\usepackage{accents}
\usepackage{graphics}
\usepackage{graphicx}
\usepackage{float}
\usepackage{url}
\usepackage{algorithm}
\usepackage{algorithmic}

\newtheorem{definition}{Definition}
\newtheorem{theorem}{Theorem}
\newtheorem{lemma}{Lemma}

\newcommand{\Z}{\mathbb{Z}}
\newcommand{\R}{\mathbb{R}}

\title{\LARGE \bf On the decoding of Barnes-Wall lattices}

\author{Vincent Corlay$^{\dagger,*}$, Joseph J. Boutros$^{\ddagger}$, Philippe Ciblat$^{\dagger}$, and Lo\"ic Brunel$^*$\\
$^{\dagger}$ Telecom Paris,Institut Polytechnique de Paris, 91120 Palaiseau, France, v.corlay@fr.merce.mee.com\\
$^{\ddagger}$ Texas A\&M University, Doha, Qatar, $^*$Mitsubishi Electric R\&D Centre Europe, Rennes, France
}

\begin{document}
\pagenumbering{gobble}

\maketitle
\thispagestyle{plain}
\pagestyle{plain}

\begin{abstract}
We present new efficient recursive decoders for the Barnes-Wall lattices based on their squaring construction. 
The analysis of the new decoders reveals a quasi-quadratic complexity in the lattice dimension and a quasi-linear complexity in the list-size. 
The error rate is shown to be close to the universal lower bound in dimensions 64 and 128. 
\end{abstract}

\section{Introduction}
Barnes-Wall ($BW$) lattices were one of the first series discovered
with an infinitely increasing fundamental coding gain~\cite{Barnes1959}. 
This series includes dense lattices in lower dimensions such as $D_4$, $E_8$, 
$\Lambda_{16}$ \cite{Conway1999}, and is deeply related to Reed-Muller codes~\cite{Forney1988}\cite{MacWilliams1977}: 
$BW$ lattices admit a Construction D based on these codes. 
Multilevel constructions attracted the recent attention of researchers,
mainly Construction C$^*$ \cite{Bollauf2019}, where lattice and non-lattice constellations are made out of binary codes.
One of the important challenges is to develop lattices with a reasonable-complexity decoding where
a fraction of the fundamental coding gain is sacrificed in order to achieve a lower kissing number. 
$BW$ lattices are attractive in this sense. 
For instance the lattice $BW_{128}$, with an equal fundamental coding gain as $Nebe_{72}$ \cite{Nebe2012}, 
sacrifices 1.5 dB of its fundamental coding gain with respect to $MW_{128}$ \cite{Elkies2001} 
while the kissing number is reduced by a factor of 200.

Several algorithms have been proposed to decode $BW$ lattices. 
Forney introduced an efficient maximum-likelihood decoding (MLD) algorithm in~\cite{Forney1988} 
for the low dimension instances of these lattices based on their trellis representation.
Nevertheless, the complexity of this algorithm is exponential in the dimension and intractable
for $n>32$: e.g. decoding in $BW_{64}$ involves $2\cdot2^{24}+2 \cdot 2^{16}$ decoders of $BW_{16}$ and decoding in $BW_{128}$ involves $2\cdot2^{48}+2\cdot 2^{32}$ decoders of $BW_{32}$ (using the two-level squaring construction to build the trellis, see \cite[Section~IV.B]{Forney1988}). 
Later, \cite{Micciancio2008}~proposed the first bounded-distance decoders (BDD) running in polynomial time:
a parallelisable decoder of complexity $O(n^2)$ and another sequential decoder of complexity $O(n\log^2(n))$.
The parallel decoder was generalized in~\cite{Grigorescu2017} to work beyond the packing radius, still in polynomial time. 
It is discussed later in the paper. 
The sequential decoder uses the $BW$ multilevel construction to perform multistage decoding:
each of the $\approx \log(n)$ levels is decoded with a Reed-Muller decoder of complexity $n\log(n)$.
This decoder was also further studied, in \cite{Harsham2013}, 
to design practical schemes for communication over the AWGN channel.
The performance of this sequential decoder is far from MLD.
A simple information-theoretic argument explains 
why multistage decoding of $BW$ lattices cannot be efficient: 
the rates of some component Reed-Muller codes
exceed the channel capacities of the corresponding levels \cite{Forney2000}\cite{Yan2013}.
As a result, no $BW$ decoders, being both practical and quasi-optimal on the Gaussian channel, 
have been designed and executed for dimensions greater than $32$.

We present new decoders for $BW$ lattices based on their $(u, u+v)$ construction \cite{MacWilliams1977}. 
We particularly consider this construction as a squaring construction \cite{Forney1988} to establish
a new recursive BDD (Algorithm~2, Section~III-A), new recursive list decoders (Algorithms~3~and~5, Sections~IV-B
and IV-C), and their complexity analysis as stated by Theorems~2-4. As an example, 
Algorithm~5 decodes $BW_{64}$ and $BW_{128}$ with a performance close to the universal lower bound 
on the coding gain of any lattice and with a reasonable complexity almost quadratic in the lattice dimension. 

\section{preliminaries}
\noindent \textbf{Lattice.} A lattice $\Lambda$ is a discrete additive subgroup of $\R^n$.
For a rank-$n$ lattice in $\R^n$, the rows of a $n\times n$ generator matrix $G$ constitute
a basis of $\Lambda$ and any lattice point $x$ is obtained via $x=zG$, where $z \in \Z^n$.
The squared minimum Euclidean distance of $\Lambda$ is $d(\Lambda)=(2\rho(\Lambda))^2$, 
where $\rho(\Lambda)$ is the packing radius.
The number of lattice points located at a distance $\sqrt{d(\Lambda)}$ 
from the origin is the kissing number $\tau(\Lambda)$.
The fundamental volume of $\Lambda$,
i.e. the volume of its Voronoi cell and its fundamental parallelotope, is denoted by
$\text{vol}(\Lambda)$.
The fundamental coding gain~$\gamma(\Lambda)$ is given by the ratio
$\gamma(\Lambda) = d(\Lambda)/\text{vol}(\Lambda)^{\frac{2}{n}}$.
The squared Euclidean distance between a point $y \in \R^n$ and a lattice point $x \in \Lambda$ 
is denoted $d(x,y)$.
Accordingly, the squared distance between $y \in \R^n$ and the closest lattice point of $\Lambda$
is $d(y,\Lambda)$. \\
For lattices, the transmission rate used with finite constellations
is meaningless. 
Poltyrev introduced the generalized
capacity \cite{Poltyrev1994}, the analog of Shannon
capacity for lattices. The Poltyrev limit corresponds to a noise variance of
$\sigma^2_{max}=\det(\Lambda)^{\frac{2}{n}}/(2 \pi e)$ 
and the point error rate is evaluated 
with respect to the distance to Poltyrev limit, i.e. $\sigma^2_{max}/\sigma^2$. \\

\vspace{-3mm}
\noindent \textbf{BDD, list-decoding, and MLD.} 
Given a lattice $\Lambda$, a radius $r>0$, and any point $y \in \R^n$,
the task of a decoder is to determine all points $x \in \Lambda$ satisfying
$d(x,y) \leq r^2(\Lambda).$
If $r < \rho(\Lambda)$, there is either no point or a unique point found and the decoder is known as BDD.
Additionally, if $d(x,y) < \rho^2(\Lambda)$, we say that $y$ is within the guaranteed error-correction radius of the lattice. 
If $r \ge \rho(\Lambda)$, there may be more than one point in the sphere.
In this case, the process is called list-decoding rather than BDD.
When list-decoding is used, lattice points within the sphere are enumerated and the decoded lattice point is the closest to $y$ among them. 
MLD simply refers to finding the closest lattice point in $\Lambda$ to any point $y \in \R^n$.
If list-decoding is used, MLD is equivalent to choosing a decoding radius equal to $R(\Lambda)$.\\

\noindent \textbf{Coset decomposition of a lattice.}
Let $\Lambda$ and $\Lambda'$ be two lattices such that $\Lambda' \subseteq \Lambda$.
If the order of the quotient group $\Lambda/\Lambda'$ is $q$, 
then $\Lambda$ can be expressed as the union of $q$ cosets of $\Lambda'$.
We denote by $[\Lambda/\Lambda']$ a system of coset representatives for this partition. 
It follows that $\Lambda = \bigcup_{x_i \in   [\Lambda/\Lambda']}\Lambda' + x_i = \Lambda'+ [\Lambda/\Lambda'].$ \\
%
\noindent \textbf{The $BW$ lattices.} Let the scaling-rotation operator $R(2n)$ in dimension $2n$ 
be defined by the application of the $2\times2$ matrix
\small
\begin{align*}
R(2) = \left[
\begin{matrix}
1 & \ \ 1 \\
1 & -1
\end{matrix}
\right]
\end{align*}
\normalsize
on each pair of components. 
I.e. the scaling-rotation operator is $R(2n)= I_{n} \otimes R(2)$, where $I_{n}$ is the $n\times n$ identity matrix and
$\otimes$ the Kronecker product.
For $\Lambda \subset \R^{2n}$ with generator matrix $G$, 
the lattice generated by $G\cdot R(2n)$ is denoted $R\Lambda$.
\begin{definition}[The squaring construction of $BW_{2n}$ \cite{Forney1988}]
The $BW$ lattices in dimension $2n$ are obtained by the following recursion:
\small
\vspace{-1mm}
\begin{align*}
\begin{split}
BW_{2n} &= \{(\underset{u_1 \in BW_n}{\underbrace{v'_1+m}},\underset{u_2 \in BW_n}{\underbrace{v'_2+m}}),
 v'_i\in RBW_{n},  m \in [BW_{n}/RBW_{n}] \},
\end{split}
\end{align*}
\normalsize
with initial condition $BW_2= \mathbb{Z}^{2}$.
\end{definition}
Using this construction, it is easily seen that $d(BW_{2n})=d(RBW_n)=2d(BW_n)$
and the fundamental coding gain increases infinitely as
 $\gamma(BW_n) = \sqrt{2} \cdot \gamma(BW_{n/2}) = \sqrt{n/2}$~\cite{Forney1988}. \\
Note that the squaring construction can be expressed under 
the form of the Plotkin $(u,u+v)$ construction \cite{Plotkin1960}:
\small
\vspace{-1mm}
\begin{align*}
BW_{2n} & = \{(v'_1+m,v'_2+m), \  v'_i \in  RBW_{n},  m \in  [BW_{n}/RBW_{n}] \},  \\
              & = \{ (v'_1+m, \underset{v'_2}{\underbrace{v'_1+v_2}}+m) \} = \{ (u_1, u_1 +v_2) \}.
\end{align*}
\normalsize
%
\vspace{-7mm}
\section{Bounded-distance $BW$ decoding}
\vspace{-1mm}
\subsection{The new BDD}
Given a point $y=(y_1,y_2)\in \mathbb{R}^{2n}$ to be decoded, 
a well-known algorithm \cite{Schnabl1995}\cite{Dumer2006} for a code 
obtained via the  $(u,u+v)$ construction is to first decode $y_1$ as $u_1$, 
and then decode $y_2-u_1$ as $v_2$ \footnote{The standard decoder for  $(u,u+v)$ has a second round:
once $v_2$ is decoded $u_1$ is re-decoded based on the two estimates $y_1$ and $y_2-v_2$.}. 
Our lattice decoder, Algorithm~\ref{algo_squar}, is double-sided 
since we also decode $y_2$ as $u_2$ and then $y_1-u_2$ as $v_2$:
the decoder is based on the squaring construction.  
The main idea exploited by the algorithm is that if there is
too much noise on one side, e.g. $y_1$,
then there is less noise on the other side, e.g. $y_2$, and vice versa.

\label{sec_BDD}
\vspace{-3mm}
\begin{algorithm}[H]
\caption{Double-sided $(u,u+v)$ decoder of $BW_{2n}$}
\label{algo_squar}
\small
\textbf{Input:} $y=(y_1,y_2) \in \R^{2n}$.
\begin{algorithmic}[1]
\STATE Decode (MLD) $y_1,y_2$ in $BW_n$ as $u_1,u_2$.
\STATE Decode  (MLD) $y_2 - u_1$ in $RBW_n$ as $v_2$. Store $\hat{x} \leftarrow (u_1, u_1+v_2)$.
\STATE Dec. (MLD) $y_1 - u_2$ in $RBW_n$ as $v_1$. Store $\hat{x}' \leftarrow (u_2+v_1,u_2)$.
\STATE \textbf{Return} $x_{dec} = \underset{x \in \{ \hat{x}, \hat{x}' \} }{\text{argmin \ \ }} ||y - x ||$
\end{algorithmic}
\normalsize
\end{algorithm}
\vspace{-5mm}
\begin{theorem}
\label{theo_BDD_BW}
Let $y$ be a point in $\mathbb{R}^{2n}$ such that $d(y,BW_{2n})$ is less than $\rho^2(BW_{2n})$.
Then, Algorithm~\ref{algo_squar} outputs the closest lattice point $x \in BW_{2n}$ to $y$.
\end{theorem}
\begin{proof}
If $(x_1, x_2) \in BW_{2n}$, then $x_1,x_2 \in BW_n$. 
Also, we have $||(y_1,y_2)||^2 = ||y_1||^2 + ||y_2||^2$.
So if $d(y,BW_{2n})<~\rho^2(BW_{2n})$, 
then at least one among the two $y_i$ is at a distance smaller
than $\frac{\rho^2(BW_{2n})}{2} = \rho^2(BW_{n})$ 
from $BW_{n}$. 
Therefore, at least one of the two $u_i$ is correct. \\
Assume (without loss of generality) that $u_1$ is correct.
We have $d(y_2-u_1,RBW_n)<\rho^2(BW_{2n}) = \rho^2(RBW_n)$. 
Therefore, $y_2-u_1$ is also correctly decoded.\\
As a result, among the two lattice points stored, 
at least one is the closest lattice point to $y$.
\end{proof}
Note that the $BW_n$ decoder in the previous proof
got exploited up to $\rho^2(BW_n)$ only. 
Consequently, Algorithm~\ref{algo_squar} should exceed
the performance predicted by Theorem~\ref{theo_BDD_BW}
given that step~1 is MLD. \\
Algorithm~\ref{algo_squar} can be generalized into the recursive Algorithm~\ref{algo_BW_rec},
where Steps 4, 5, and 6 of the latter algorithm 
replace Steps 1, 2, and 3 of Algorithm~\ref{algo_squar}, respectively.
This algorithm is similar to the parallel decoder of~\cite{Micciancio2008}.
The main difference is that~\cite{Micciancio2008} 
uses the automorphism group of $BW_{2n}$ to get four candidates at each recursion whereas we use 
the squaring construction to generate only two candidates.
Nevertheless, both our algorithm and \cite{Micciancio2008}
use four recursive calls at each recursive section and have the same asymptotic complexity.
\begin{algorithm}
\caption{Recursive BDD of $BW_{2n}$ (where $2n=2^t$)}
\label{algo_BW_rec}
\small
\textbf{Function} $RecBW(y,t)$ \\
\textbf{Input:} $y=(y_1,y_2)  \in \mathbb{R}^{2^t}$, $1 \leq t$. \\
\begin{algorithmic}[1]
\vspace{-3mm}
\IF{t = 1}
\STATE $x_{dec} \leftarrow (\lfloor y_1 \rceil, \lfloor y_2 \rceil)$ \ \ \  // Decoding in $\mathbb{Z}^2$
\ELSE
\STATE $u_1 \leftarrow RecBW(y_1,t-1)$, $u_2 \leftarrow RecBW(y_2,t-1)$ 

\vspace{1mm}

// $y_2 - u_1$ (and $y_1 - u_2$) should be decoded in $RBW_n$: \\
// this is equivalent to decoding  $(y_2 - u_1)\cdot R(2^{t-1})^{-1}$ in $BW_n$ \\
// and then rotate the output lattice point by $R(2^{t-1})$.
\STATE $v_2 \leftarrow RecBW((y_2 - u_1)\cdot R(2^{t-1})^T/2,t-1) \cdot R(2^{t-1}) $. \\
Store $\hat{x} \leftarrow (u_1, u_1+v_2)$.
\STATE $v_1 \leftarrow RecBW((y_1 - u_2) \cdot R(2^{t-1})^T/2,t-1) \cdot R(2^{t-1})$. \\
 Store $\hat{x}' \leftarrow  (u_2+v_1,u_2)$.

\vspace{1mm}

\STATE $x_{dec} = \underset{x \in \{ \hat{x}, \hat{x}' \} }{\text{argmin \ \ }} ||y - x ||$
\ENDIF
\STATE \textbf{Return} $x_{dec}$
\end{algorithmic}
\normalsize
\end{algorithm}
\vspace{-2mm}
\begin{theorem}
\label{theo_complex_1}
Let $n$ be the dimension the lattice $BW_n$ to be decoded.
The complexity of Algorithm~\ref{algo_BW_rec} is $O(n^2)$.
\end{theorem}
\begin{proof}
Let $\mathfrak{C}(n)$ be the complexity of the algorithm for $n=2^t$.
We have $\mathfrak{C}(n) = 4 \mathfrak{C}(n/2)+O(n) = O(n^2).$
\end{proof}

\subsection{Performance on the Gaussian channel}

In the appendix (see Section~\ref{app_1}), 
we show via an analysis of the effective error coefficient of Algorithm~\ref{algo_BW_rec} that the loss in performance 
of this algorithm compared 
to MLD (in dB) is expected to grow linearly with $n$. \\
Our simulations show that there is a loss of $\approx$0.25 dB for $n=16$, $\approx$0.5 dB for $n=32$, $\approx$1.25 dB for $n=64$ 
(compare $\aleph=1$ and $\aleph=20$ on Figure~\ref{fig_influ_list}) and $\approx$2.25 dB for $n=128$.
As a result, this BDD is not suited for effective decoding of $BW$ lattices on the Gaussian channel.
However, it is essential for building efficient decoders as shown in the next section.


\section{List-decoding of $BW$ lattices beyond the packing radius}

Let $L(\Lambda,r^2)$ be the maximum number of lattice points of $\Lambda$ within a sphere of radius $r$ 
around any $y \in \R^n$. If $\Lambda=BW_n$ we write $L(n,r^2)$. 
The following lemma is proved in \cite{Grigorescu2017}.
\vspace{-2mm}
\begin{lemma}
\label{lemma_johnson}
The list size of the $BW_n$ lattices is bounded as~\cite{Grigorescu2017}:  
\begin{itemize}
\item $L(n,r^2) \leq \frac{1}{2 \epsilon }$ if $r^2 \leq d(BW_n)(1/2 - \epsilon)$, $0<\epsilon \le 1/4$.
\item $L(n,r^2) = 2n$ if $r^2 = d(BW_n)/2$.
\item $L(n,r^2) \le   2  n^{16 \log_2(1/\epsilon)}$ if $r^2 \leq d(BW_n)(1 - \epsilon)$, $0~<~\epsilon~\le~1/2$.
\end{itemize}
\end{lemma}

\cite{Grigorescu2017} also shows that the parallel BDD of \cite{Micciancio2008}, which uses 
the automorphism group of $BW_n$,
can be slightly modified to output a list of all lattice points lying at a squared distance $r^2 =d(BW_n)(1 - \epsilon), \forall \epsilon >0$,  
from any $y\in \R^n$ in time $O(n^2) \cdot L(n,r^2)^2$. 
With Lemma~\ref{lemma_johnson}, this becomes $n^{O(\log(1/\epsilon))}$. 
This result is of theoretical interest: 
it shows that there exists a polynomial time algorithm in the dimension
for any radius bounded away from the minimum distance.
However, the quadratic dependence in the list-size is a drawback:
finding an algorithm with quasi-linear dependence in the list-size is stated as an open problem in \cite{Grigorescu2017}.

In the following, we show that if we use the squaring construction rather 
than the automorphism group of $BW_n$ for list-decoding,
it is possible to get a quasi-linear complexity in the list size. 
This enables to get a practical list-decoding algorithm up to $n=128$.

\subsubsection{Some notations}

Notice that $L(n,r^2)=L(RBW_n, 2 r^2)$, e.g. 
both are equal to $2n$ if $r^2=d(BW_n)/2$.
It is therefore convenient to consider the relative squared distance 
as in~\cite{Grigorescu2017}: $\delta(x,y) = \frac{d(x,y)}{d(\Lambda)}$, $x \in \Lambda$ \footnote{Here, $\Lambda$ should be the ``smallest" lattice to which $x$ belongs: e.g. take $u \in BW_n$. We also have $u \in RBW_n$ but $\Lambda$ should be $BW_n$. }.
Then, if we define $l(\Lambda, r^2 / d(\Lambda))=L(\Lambda,r^2)$
this yields for instance $l(n,1/2)=l(BW_n,1/2)=l(RBW_n,1/2)=2n$. 
The relative squared radius is defined as the quantity $r^2 / d(\Lambda)$.
For the rest of this section, $\delta$ is the relative squared radius considered for decoding.
Let $y=(y_1,y_2) \in \R^{2n}$ and $x=(u_1,u_1+v_2) = (u_2+v_1,u_2) \in BW_{2n}$ 
be any lattice point where $\delta(x,y) \le \delta$.
We recall that for BDD of $BW_n$ we have $\delta=1/4$. \\
The following lemma is trivial, but convenient to manipulate distances.
\begin{lemma}(Lemma~2.1 in~\cite{Grigorescu2017}) \\
Let $y=(y_1,y_2) \in \R^{2n}$ and $x=(u_1,u_1+v_2) \in~BW_{2n}$. Then,
\begin{align}
\label{equ_dist_conv}
\delta(x,y) = \delta(u_1,y_1)/2 + \delta(v_2,y_2-u_1).
\end{align}
\end{lemma}

\subsubsection{List-decoding with $r^2 < 9/16 d(BW_{n})$ }
Assume that the squared norm of the noise is $r^2$ and $\delta = (r^2+\epsilon)/d(BW_{n})$.
Consider $d(x,y)=d(u_1,y_1)+d(u_2,y_2)$.
We split the possible situations into two main cases (similarly to Steps 2-3 of Algorithm~\ref{algo_squar}): $d(u_1,y_1) \leq r^2/2$ and $d(u_1,y_1) > r^2/2$.
For the first case, $y_1$ should be list-decoded in $BW_{n/2}$, and for each $u_1 \in BW_{n/2}$ in the resulting list, $y_2-u_1$ should be list-decoded in $RBW_{n/2}$.
Regarding the noise repartition, one can get the following two extreme configurations (but not simultaneously): $d(u_1,y_1)=r^2/2$, i.e. $\delta(u_1,y_1)=\delta$ and $d(v_2,y_2-u_1)=r^2$, i.e. $\delta(v_2,y_2-u_1)=\delta$.
Consequently, without any advanced strategy, the relative squared decoding radius to list-decode in $BW_{n/2}$ and $RBW_{n/2}$ should be $\delta$. 
The maximum of the product of the two resulting list-sizes, which is a key element in the complexity analysis below, is $l(n/2, \delta)^2$.
In order to reduce this number, we split this first case (i.e. $d(u_1,y_1)\leq r^2/2$) into two sub-cases. Let $0 \leq a' \leq r^2/2$.
\begin{itemize}
\item $0 \leq d(u_1,y_1) < a' $ and $r^2/2  < d(v_2,y_2-u_1) \le r^2$: 
then, $y_1$ should be list-decoded in $BW_{n/2}$ with a relative squared radius $a_1=a'/d(BW_{n/2})$ 
and $y_2-u_1$ list decoded in $RBW_{n/2}$ with a relative squared radius $\delta$. 
\item $a' \le d(u_1,y_1) \leq r^2/2$ and $r^2/2  < d(v_2,y_2-u_1) \le r^2-a'$: 
then,  $y_1$ should be list-decoded in $BW_{n/2}$
with a relative squared radius $\delta$ 
and $y_2-u_1$ list-decoded in $RBW_{n/2}$ with a relative squared radius $a_2=(r^2-a')/d(RBW_{n/2})$. 
\end{itemize} 
The size of the two resulting lists are bounded by $l(n/2,a_1)\cdot l(RBW_{n/2},\delta)$ and $l(n/2,\delta)\cdot l(RBW_{n/2},a_2)$.
Consequently, if we choose $a_1=a_2=a$, i.e.  $a=2/3 \delta$,
the two bounds are equal. The maximum number of candidates to consider becomes $2  l(n/2,\delta)\cdot l(n/2,a)$ which is likely to be much smaller than $l(n/2, \delta)^2$, the bound obtained without the splitting strategy. The second case (i.e. $d(u_1,y_1) > r^2/2$) is identical by symmetry.


This analysis yields Algorithm~\ref{algo_BW_rec_list} listed below.
The ``removing step" (10 in bold) is added to ensure 
that a list with no more than $l(n,\delta)$ elements 
is returned by each recursive call.
The maximum number of points to process by this removing step is $4 l(n/2,\delta) l(n/2,a)$. 
Regarding Step 11, using the classical Merge Sort algorithm, 
it can be done in \\ $O(n \cdot l(n,\delta) \log(l(n,\delta) ) )$ operations (see App.~\ref{app_2}). 

The following theorem shows that we get an algorithm of quasi-linear complexity in the list size $l(n,\delta)$. 
\begin{theorem}
\label{main_theo}
Given any point $y \in \R^n$ and $1/4\leq \delta < 9/16$, 
Algorithm~\ref{algo_BW_rec_list} outputs the list of all lattice points in $BW_n$ lying 
within a sphere of relative squared radius $\delta$ around $y$ in time:
\begin{itemize}
\small
\item $O(n^2\log n) l(n,\delta) + O(n \log n) l(n, \delta) \log l(n, \delta)$ \\ if $1/4\leq \delta < 3/8$
\item $O(n^{1+\log_2(1+l(n,2/3 \delta ))} \log n) l(n,\delta) l(n,2/3 \delta) +$  \\
$O(n^{\log_2(1+l(n,2/3 \delta ))})  l(n,\delta) \log l(n,\delta)$ if $3/8 \leq \delta < 9/16$.
\end{itemize}
\end{theorem}

\begin{proof}
Let $\mathfrak{C}(n,\delta)$ be the complexity of Algorithm~\ref{algo_BW_rec_list}.
We have
\small
\begin{align}
\label{eq_main}
\begin{split}
\mathfrak{C}(n,\delta)  & \leq  \underset{\text{Numb. recursive calls with }\delta}{\underbrace{(2+2l(n/2,a)) \mathfrak{C}(n/2,\delta) }}   
+ \underset{\text{Numb. recursive calls with }a}{\underbrace{(2+2l(n/2,\delta) ) \mathfrak{C}(n/2,a)}}+  \\
&\underset{\text{removing}}{\underbrace{4 l(n/2,\delta)l(n/2,a) O(n)}} + \underset{\text{Merge Sort}}{\underbrace{O(n \cdot l(n,\delta) \log(l(n,\delta)) )}}.
\end{split}
\end{align}
\normalsize


If $\delta < 3/8$, then $l(n,a) \leq 1$ and $\mathfrak{C}(n/2,a) \leq  \mathfrak{C}(n/2,1/4) = O((n/2)^2)$ (the complexity of Algorithm~\ref{algo_BW_rec}).
\eqref{eq_main} becomes 
\small
\begin{align}
\label{eq_inte_res}
\begin{split}
\mathfrak{C}(n,\delta)  = 
O(n^2\log n) l(n,\delta) +
O(n \log n) l(n, \delta) \log l(n, \delta).
\end{split}
\end{align}
\normalsize
If $\delta < 9/16$, then $\mathfrak{C}(n/2,a) <  \mathfrak{C}(n/2,3/8)$. 
Using \eqref{eq_inte_res} for $\mathfrak{C}(n/2,3/8)$, \eqref{eq_main} becomes 
\begin{align*}
\begin{split}
&\mathfrak{C}(n,\delta)  = 
O(n^{1+\log_2(1+l(n,2/3 \delta ))} \log n) l(n,\delta) l(n,2/3 \delta) + \\
 & O(n^{\log_2(1+l(n,2/3 \delta ))})  l(n,\delta) \log l(n,\delta) 
\end{split}
\end{align*}
\normalsize

\end{proof}
\vspace{-3mm}
\begin{algorithm}
\caption{First recursive list-decoding of $BW_{2n}$  ($2n=2^t$).}
\label{algo_BW_rec_list}
\small
\textbf{Function} $ListRecBW(y,t,\delta)$ \\
\textbf{Input:} $y=(y_1,y_2) \in \mathbb{R}^{2^t}$, $1 \leq t$, $1/4 \le \delta<3/4$. \\
\begin{algorithmic}[1]
\vspace{-3mm}
\STATE $a \leftarrow 2/3 \cdot  \delta$
\STATE $r \leftarrow \sqrt{2^{t-1} \cdot \delta}$
\IF{t = 1}
\STATE $\hat{x} \leftarrow  Enum_{\Z_2}(y,r)$\ \ // Enum. in $\mathbb{Z}^2$ with radius $r=\sqrt{\delta}$.
\ELSE


\STATE $\hat{x}_1 \leftarrow SubRoutine(y_1,y_2,t,a,\delta,0)$

\STATE $\hat{x}_2 \leftarrow SubRoutine(y_1,y_2,t,\delta,a,0)$

\STATE $\hat{x}_3 \leftarrow SubRoutine(y_2,y_1,t,\delta,a,1)$

\STATE $\hat{x}_4 \leftarrow SubRoutine(y_2,y_1,t,a,\delta,1)$

\vspace{1mm}
\STATE \textbf{Remove all candidates at a distance $> r$ from $y$}.
\vspace{1mm}

\STATE Sort the remaining list of candidates in a lexicographic order and remove all duplicates.

\vspace{1mm}
\ENDIF
\STATE \textbf{Return} the list of all the candidates remaining.
\end{algorithmic}
\normalsize
\end{algorithm}

\begin{algorithm}
\caption{Subroutine of Algorithms~\ref{algo_BW_rec_list}}
\label{algo_BW_subr}
\small
\textbf{Function} $SubRoutine(y_1,y_2,t,\delta_1,\delta_2,reverse)$ \\
\textbf{Input:} $y_1,y_2 \in \mathbb{R}^{2^{t-1}}$, $1 \leq t$, $0 < \delta_1,\delta_2 \leq 3/4$, $rev.~\in~\{ 0,1 \}$. \\
\begin{algorithmic}[1]
\vspace{-3mm}
\IF{$\delta_1 \leq 1/4$}
\STATE $u_1\_List \leftarrow RecBW(y_1,t-1)$
\ELSE
\STATE $u_1\_List \leftarrow ListRecBW(y_1,t-1,\delta_1)$
\ENDIF
\vspace{1mm}
\FOR{$u_1 \in u_1\_List$}
\IF{$\delta_2 \leq 1/4$}
\STATE $v_2\_List(u_1)  \leftarrow RecBW((y_2~-~t_1)~\cdot R(2^{t-1})^T/2,
t~-~1) \cdot R(2^{t-1}) $
\ELSE
\STATE $v_2\_List(u_1) \leftarrow ListRecBW((y_2-t_1)~\cdot R(2^{t-1})^T/2,
t-1, \delta_2 ) \cdot R(2^{t-1})$
\ENDIF
\ENDFOR
\vspace{1mm}
\FOR{$u_1 \in u_1\_List$}
\FOR{$v_2 \in v_2\_List(u_1)$}
\IF{$reverse=0$}
\STATE Compute and store $\hat{x} \leftarrow (u_1, v_2+u_1)$.
\ELSE
\STATE Compute and store  $\hat{x} \leftarrow (v_2+u_1,u_1)$.
\ENDIF
\ENDFOR
\ENDFOR
\vspace{1mm}
\STATE \textbf{Return} the list of all candidates $\hat{x}$.
\vspace{1mm}

\end{algorithmic}
\normalsize
\end{algorithm}
Unfortunately, the performance of Algorithm~\ref{algo_BW_rec_list} 
on the Gaussian channel is disappointing.
This is not surprising: notice that due to the ``removing step" (10 in bold), 
some points that are correctly decoded by Algorithm~\ref{algo_BW_rec} 
(the BDD) are not in the list outputted by Algorithm~\ref{algo_BW_rec_list}!
Therefore, instead of removing all candidates at a distance greater than $r$, 
it is tempting to keep $\aleph$ candidates at each step.

\subsection{An efficient list decoder on the Gaussian channel}
We call Algorithm~5 a modified version 
of Algorithm~\ref{algo_BW_rec_list} where the $\aleph(\delta)$ closest candidates are kept at each recursive step (instead of step 10, i.e. keeping only the points in the sphere of radius $r$) and steps 10 and 11 are flipped.
The size of the list $\aleph(\delta)$, for a given $\delta$, is a parameter to be fine tuned:
e.g. for $\delta = 3/8$, one needs to choose only $\aleph(3/8)$ but for $\delta = 1/2$, one needs to choose $\aleph(1/2)$ and $\aleph(2/3 \cdot 1/2=1/3)$. 
The following theorem follows from Theorem~\ref{main_theo}.
\begin{theorem}
\label{theo_main_pract}
Given any point $y \in \R^n$ and $1/4\leq \delta < 9/16$, 
Algorithm~5 outputs the list of all lattice points in $BW_n$ lying 
within a sphere of relative squared radius $\delta$ around $y$ in time:
\begin{itemize}
\small
\item $O(n^2\log n) \aleph(\delta)+ O(n \log n)\aleph(\delta) \log \aleph(\delta)$ \\
$1/4\leq \delta < 3/8$
\item $O(n^{1+\log_2(1+\aleph(2/3\delta ))} \log n) \aleph(\delta ) \aleph(2/3\delta )+$  \\
$O(n^{\log_2(1+\aleph(2/3\delta ))}) \aleph(\delta ) \log \aleph(\delta )$ 
if $3/8 \leq \delta < 9/16$.
\end{itemize}
\end{theorem}
Note that on the Gaussian channel, the probability that $y$ is exactly between two lattice points is 0. As a result, we can assume that $l(n,1/4)=1$ and thus include $\delta=3/8$ in the first case of Theorem~\ref{theo_main_pract}.

As comparison, a similar modification of the algorithm of \cite{Grigorescu2017} would yield a complexity $O(n^2) \cdot \aleph(\delta)^2$. In the next section on simulations, we show that for $n=32$ and $n=64$, one should choose $\delta=3/8$ and $\aleph(\delta)$ should be at least 10 and 20, respectively. Moreover, for $n=128$ quasi-optimal performance is achieved for $\delta=1/2$ with $\aleph(\delta)=1000$ and $\aleph(2/3 \delta)=4$.
With these parameters, our algorithm has a clear advantage thanks to the quasi-linear dependence in $\aleph(\delta)$.  \\
\section{Numerical results}

\subsection{Performance of Algorithm~5}

Figure~\ref{fig_influ_list} shows the influence of the list size 
when decoding $BW_{64}$ 
using Algorithm~5 with $\delta=3/8$. 
On this figure we also plotted an estimate of the MLD performance of $BW_{64}$, 
obtained as $\tau(BW_{64})/2 \cdot \text{erfc}( \gamma/(8  \sigma^2_{max}/\sigma^2))$ \cite[Chap.~3]{Conway1999}.

Figure~\ref{fig_perf_list} depicts the performance 
of Algorithm~5 for the $BW$ lattices up to $n=128$
and the universal bounds provided in \cite{Tarokh1999} (see also \cite{Forney2000} or \cite{Ingber2013}, 
where it is called the sphere lower bound).
These universal bounds are limits on  the highest possible coding gain using $any$ lattice in $n$ dimensions.
For each $BW_n$ we tried to reduce as much as possible the list size while keeping quasi-MLD performance.
The choice of $\delta=3/8$ yields quasi-MLD performance up to $n=64$ with 
small list size and thus reasonable complexity.
This shows that $BW_{64}$, with Algorithm~5, is a good
candidate to design finite constellations in dimension 64.
However, for $n=128$ one needs to set $\delta=1/2$ and choose $\aleph(\delta)=1000$.
Nevertheless, $\aleph(2/3 \cdot \delta)$ can be as small as 4, which is still tractable. \\
\begin{figure}
\centering
\vspace{-5mm}
\includegraphics[width=0.8\columnwidth]{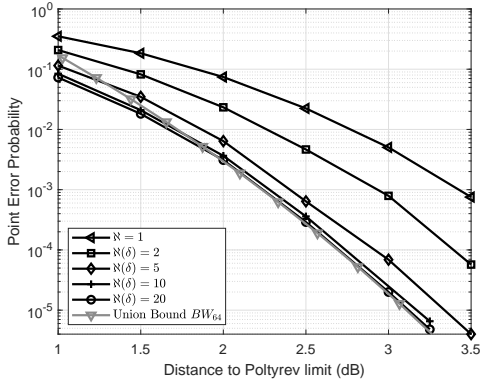}
\vspace{-2mm}
\caption{Influence of the list size when decoding $BW_{64}$ using Alg.~5, $\delta=3/8$.}
\label{fig_influ_list}
\end{figure}
\begin{figure}
\centering
\vspace{-4mm}
\includegraphics[width=0.8\columnwidth]{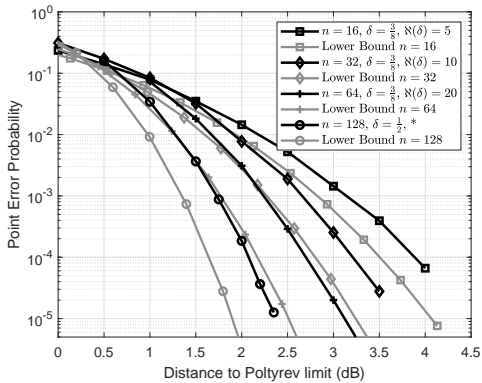}
\vspace{-2mm}
\caption{Algorithm~5 for the $BW$ lattices up to $n~=~128$
and the universal bounds of \cite{Tarokh1999}. $^*$For $n=128$, $\aleph(\delta)=1000$ and $\aleph(2/3 \delta)=4$.}
\label{fig_perf_list}
\vspace{-4mm}
\end{figure}
\indent We compare these performances with existing schemes at $P_e=10^{-5}$.
For fair comparison between the dimensions, we let $P_e$ be either the normalized error probability,
which is equal to the point error-rate divided by the dimension (as done in e.g. \cite{Tarokh1999}),
or the symbol error-rate.\\
First, several constructions have been proposed for block-lengths around $n=100$ in the literature.
In \cite{Matsumine2018} a two-level construction based on BCH codes with $n=128$ achieves
this error-rate at 2.4 dB. The decoding involves an OSD of order 4 with 1505883 candidates.
In \cite{Agrawal2000} the multilevel (non-lattice packing) $\mathcal{S}_{127}$ ($n=127$) has similar
performance but with much lower decoding complexity via generalized minimum distance decoding. 
In \cite{Sakzad2010} a turbo lattice with $n=102$ and in \cite{Sommer2008} a LDLC with $n=100$
achieve the error-rate with iterative methods at respectively 2.75 dB and 3.7 dB (unsurprisingly, these two schemes are efficient for larger block-lengths).
All these schemes are outperformed by $BW_{64}$, where $P_e=10^{-5}$ is reached at 2.3 dB.
Moreover, $BW_{128}$ has $P_e=10^{-5}$ at $1.7$ dB, which is similar to many schemes with block-length $n=1000$
such as the LDLC (1.5 dB) \cite{Sommer2008},  the turbo lattice (1.2 dB) \cite{Sakzad2010}, 
the polar lattice with $n=1024$ (1.45 dB) \cite{Yan2013}, and the LDA lattice (1.27 dB) \cite{diPietro2012}. 
This benchmark is summarized in Figure~\ref{fig_benchmark}.
\begin{figure}
\centering
\vspace{-5mm}
\includegraphics[width=0.8\columnwidth]{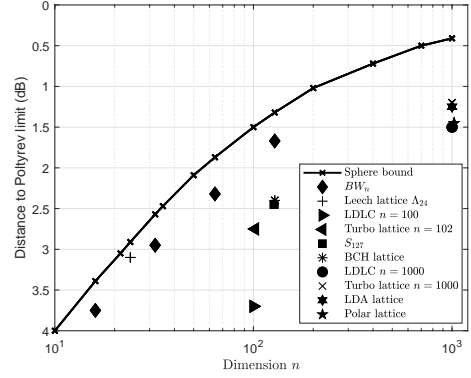}
\vspace{-2mm}
\caption{Perf. of different lattices for normalized error probability $P_e=10^{-5}$.}
\label{fig_benchmark}
\vspace{-4mm}
\end{figure}

\subsection{Performance of $BW$ finite constellations}
We uncover the performance of a Voronoi constellation \cite{Conway1983}\cite{Forney1989_2} 
based on the partition $BW_{64}/2^{\eta} BW_{64}$ via Monte Carlo simulation, 
where $\eta$ is the desired rate in bits per channel use (bpcu): i.e. both the coding lattice and the shaping
lattice are based on $BW_{64}$. It follows that the encoding complexity is the same as the
decoding complexity: the complexity of Algorithm~5 with $\delta=3/8$ and $\aleph(\delta)=20$.
Figure~\ref{fig_finite_constell} exhibits the performance of our scheme for $\eta=4$ bpcu.
In our simulation, the errors are counted on the uncoded symbols. 
The error-rate also includes potential errors due to incomplete encoding, 
which seem to be negligible compared to decoding errors.
Again, we plotted the best possible performance of $any$ lattice-based constellation in dimension 64 (obtained from \cite{Tarokh1999}).
The scheme performs within 0.7 dB of the bound.
\begin{figure}
\centering
\includegraphics[width=0.8\columnwidth]{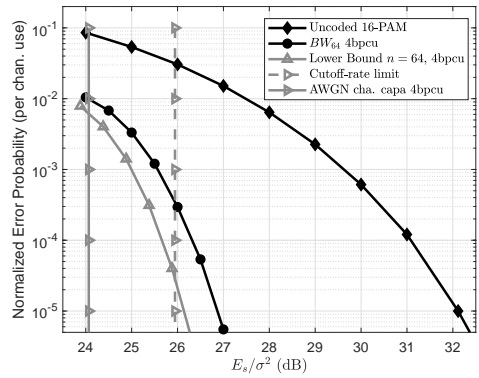}
\vspace{-2mm}
\caption{Performance of a Voronoi constellation based on the partition $BW_{64}/2^{4} BW_{64}$ where 
Algorithm~5, with $\delta=3/8$ and $\aleph(\delta)=20$, is used for encoding and decoding.
The cutoff-rate limit is 1.7+0.179 dB right to Shannon limit (coding + shaping loss for $n=64$)\cite{Forney1998}.}
\label{fig_finite_constell}
\vspace{-4mm}
\end{figure}
\vspace{-1mm}
\section{Conclusions}
\vspace{-1mm}
Our recursive paradigm can be seen as a tree search algorithm
and our decoders fall therefore in the class of sequential decoders.
While the complexity of Algorithm~5 remains stable and low for $n\le 64$,
there is a significant increase for $n=128$ 
and it becomes intractable for $n=256$ due to larger lists. 
This is not surprising from the cut-off rate perspective \cite{Forney1998};
For $n=64$ the MLD is still at a distance of $1$ dB from this limit 
(Figure~\ref{fig_finite_constell}), but it is very close to the limit for $n=128$
and potentially better at larger $n$. 
One should not expect to perform quasi-MLD of these lattices
with any sequential decoder.
This raises the following open problem: 
can we decode lattices beyond the cut-off rate in non-asymptotic dimensions, i.e. $n < 300$, where classical
capacity-approaching decoding techniques (e.g. BP) cannot~be~used?

\section{Appendix}
\subsection{Analysis of the effective error coefficient}
\label{app_1}
Let us define the decision region of a BDD algorithm $\mathbb{R}_{BDD}(\textbf{0})$ as 
the set of all points of the space that are decoded to \textbf{0} by the algorithm.
The number of points at distance $\rho(\Lambda)$ from the origin that are not necessarily 
decoded to $\textbf{0}$  are called boundary point of $\mathbb{R}_{BDD}(\textbf{0})$
The number of such points is called effective error coefficient of the algorithm. 
The performance of BDD algorithms are usually estimated via  
this effective error coefficient \cite{Forney1996}\cite{Salomon2006}.
Indeed, BDD up to the packing radius achieves the best possible error exponent on the Gaussian channel,
but the performance might be significantly degraded, 
compared to MLD, due to a high effective error coefficient. \\
In \cite{Micciancio2008}, the error coefficient of the parallel decoder is not computed 
and the performance of the algorithm is not assessed on the Gaussian channel.
The following analysis of Algorithm~\ref{algo_BW_rec} is also valid for the parallel decoder \cite{Micciancio2008}.
Let us express the point to be decoded as $y=x+\eta$, where $x \in BW_{n}$ and $\eta$ is a noise pattern.
Scale $BW_{n}$ such that its packing radius is 1. It is easily seen that
any $\eta$ of the form $(\pm \frac{1}{\sqrt{2^t}}^{2^t})=(\pm \frac{1}{\sqrt{2^t}}, ... , \pm \frac{1}{\sqrt{2^t}})$, $t=\log_{2}(n)$, is on
the boundary of $\mathbb{R}_{BDD}(\textbf{0})$.
The number of such noise patterns is $2^{2^t}=2^{n}$. 
According to Forney's rule of thumb, every factor-of-two increase in the number
of nearest neighbor results in a $0.2$ dB loss in effective coding gain \cite{Forney1998}.
Since the kissing number of $BW_n$ is 
$\prod_{i=1}^t (2^i+2) \approx 4.768 ...  \cdot 2^{0.5 log_{2} n (log_2 n +1)}$ \cite{Conway1999}, 
to be compared to the above number of noise patterns $2^n$,
we see that the loss in performance compared to MLD (in dB) is expected to grow as $\approx 0.2n$. 
However, this rule holds only if the effective error coefficient is not too large and the performance of Algorithm~\ref{algo_BW_rec} is not as bad in practice.
Nevertheless, this analysis hints that one should expect the performance of this BDD to degrade as $n$ increases. 

\subsection{The Merge Sort Algorithm}
\label{app_2}

Let $l^{k,n}=(x_1^n, x_{2}^n ..., x_{k}^n)$ be a list of $k$ elements $x$ of dimension $n$ (assume for the sake of simplicity that $k$ is a power of 2).
This list can be split into two lists of equal size $l_1^{k/2,n}$ and $l_2^{k/2,n}$ and we write
$l^{k,n}=(l_1^{k/2,n},l_2^{k/2,n})$. \\
Then, we define the function $Merge$ as a function that takes two sorted lists of $k$ elements $x$ of dimension $n$ as input (as well as $k$ and $n$) and returns a unique sorted list of the $2k$ elements. There exists several variants of this function, but the complexity is always $O(n \cdot k)$.
\begin{algorithm}[H]
\caption{Merge Sort Algorithm}
\textbf{Function:} $MS(l^{k,n},k,n)$ \\
\textbf{Input: $l^{k,n}=(l_1^{k/2,n},l_2^{k/2,n})$, $k\geq1 $, $ n~\geq~1$.} 
\begin{algorithmic}[1]
\IF{$k=1$}
\STATE \textbf{Return} $l^{k,n}$.
\ELSE 
\STATE  \textbf{Return} $Merge(MS (l_1^{k/2,n},\frac{k}{2},n),M(l_2^{k/2,n},\frac{k}{2},n),k,n)$
\ENDIF
\end{algorithmic}
\normalsize
\end{algorithm}

Let $\mathfrak{C}(k,n)$ be the complexity of the $MS$ function (Algorithm~6).
The complexity of this algorithm is 
$\mathfrak{C}(k,n) =2\mathfrak{C}(k/2,n) +O(k~\cdot~n) = O(k \log(k) \cdot n)$



%
%
%


\end{document}